\documentclass[letterpaper,11pt]{article}

\usepackage{pf2}
\pflongnumbers

\let\oldstepref\stepref
\renewcommand{\stepref}[1]{step \oldstepref{#1}}

\usepackage{amsmath,amsthm,mathtools}
\usepackage{amssymb}
\usepackage{amsfonts}
\usepackage{mathrsfs}

\usepackage{multirow}
\usepackage{array}

\usepackage{enumerate}
\usepackage{hyperref} % ref online resource
\usepackage{graphicx} % includes graphics
\usepackage{float} % includes graphics
\graphicspath{{./figures/}}

\usepackage{color, verbatim}
\usepackage[usenames,dvipsnames]{xcolor}

\usepackage{cleveref}

\usepackage{epsfig}
\usepackage{subfig}
\usepackage{epstopdf}
\usepackage{xargs}

\usepackage{hhline}
\usepackage[linesnumbered]{algorithm2e}
\crefname{algocf}{algo.}{algo(s)}

\allowdisplaybreaks % allow pagebreaks in multiline eqs

\newcommand\nc\newcommand
\nc{\bb}[1]{\mathbb{#1}}
\renewcommand{\cal}[1]{\mathcal{#1}}
\nc{\mbf}[1]{\mathbf{#1}}
\nc\defeq{\triangleq}

\nc\Defeq[1]{\stackrel{\normalfont\mbox{#1}}{=} }

\newtheorem{theorem}{Theorem}

\crefname{definition}{defn.}{defns}
\newtheorem{lemma}[theorem]{Lemma}
\newtheorem*{lemma*}{\indent Lemma}

\newtheorem*{remark}{\indent Remark}

\DeclarePairedDelimiter{\set}{\lbrace}{\rbrace}
\DeclarePairedDelimiter{\br}{\lparen}{\rparen}

\DeclarePairedDelimiter{\abs}{\lvert}{\rvert}
\nc\bfa{{\bf{a}}}\nc\bfA{{\boldsymbol A}}\nc\cA{{\cal A}} \nc\fA[1]{A\br*{#1}} \nc\fa[1]{a\br*{#1}}  \nc\rmA{\mathrm{A}} \nc\rma{\mathrm{a}}
\nc\bfb{{\bf{b}}}\nc\bfB{{\boldsymbol B}}\nc\cB{{\cal B}} \nc\fB[1]{B\br*{#1}} \nc\fb[1]{b\br*{#1}}  \nc\rmB{\mathrm{B}} \nc\rmb{\mathrm{b}}
\nc\bfc{{\bf{c}}}\nc\bfC{{\boldsymbol C}}\nc\cC{{\cal C}} \nc\fC[1]{C\br*{#1}} \nc\fc[1]{c\br*{#1}}  \nc\rmC{\mathrm{C}} \nc\rmc{\mathrm{c}}
\nc\bfd{{\bf{d}}}\nc\bfD{{\boldsymbol D}}\nc\cD{{\cal D}} \nc\fD[1]{D\br*{#1}} \nc\fd[1]{d\br*{#1}}  \nc\rmD{\mathrm{D}} \nc\rmd{\mathrm{d}}
\nc\bfe{{\bf{e}}}\nc\bfE{{\boldsymbol E}}\nc\cE{{\cal E}} \nc\fE[1]{E\br*{#1}} \nc\fe[1]{e\br*{#1}}  \nc\rmE{\mathrm{E}} \nc\rme{\mathrm{e}}
\nc\bff{{\bf{f}}}\nc\bfF{{\boldsymbol F}}\nc\cF{{\cal F}} \nc\fF[1]{F\br*{#1}} \nc\ff[1]{f\br*{#1}}  \nc\rmF{\mathrm{F}} \nc\rmf{\mathrm{f}}
\nc\bfg{{\bf{g}}}\nc\bfG{{\boldsymbol G}}\nc\cG{{\cal G}} \nc\fG[1]{G\br*{#1}} \nc\fg[1]{g\br*{#1}}  \nc\rmG{\mathrm{G}} \nc\rmg{\mathrm{g}}
\nc\bfh{{\bf{h}}}\nc\bfH{{\boldsymbol H}}\nc\cH{{\cal H}} \nc\fH[1]{H\br*{#1}} \nc\fh[1]{h\br*{#1}}  \nc\rmH{\mathrm{H}} \nc\rmh{\mathrm{h}}
\nc\bfi{{\bf{i}}}\nc\bfI{{\boldsymbol I}}\nc\cI{{\cal I}} \nc\fI[1]{I\br*{#1}} \nc\rmI{\mathrm{I}} \nc\rmi{\mathrm{i}}
\nc\bfj{{\bf{j}}}\nc\bfJ{{\boldsymbol J}}\nc\cJ{{\cal J}} \nc\fJ[1]{J\br*{#1}} \nc\fj[1]{j\br*{#1}} \nc\rmJ{\mathrm{J}} \nc\rmj{\mathrm{j}}
\nc\bfk{{\bf{k}}}\nc\bfK{{\boldsymbol K}}\nc\cK{{\cal K}} \nc\fK[1]{K\br*{#1}} \nc\fk[1]{k\br*{#1}} \nc\rmK{\mathrm{K}} \nc\rmk{\mathrm{k}}
\nc\bfl{{\bf{l}}}\nc\bfL{{\boldsymbol L}}\nc\cL{{\cal L}} \nc\fL[1]{L\br*{#1}} \nc\fl[1]{l\br*{#1}} \nc\rmL{\mathrm{L}} \nc\rml{\mathrm{l}}
\nc\bfm{{\bf{m}}}\nc\bfM{{\boldsymbol M}}\nc\cM{{\cal M}} \nc\fM[1]{M\br*{#1}} \nc\fm[1]{m\br*{#1}} \nc\rmM{\mathrm{M}} \nc\rmm{\mathrm{m}}
\nc\bfn{{\bf{n}}}\nc\bfN{{\boldsymbol N}}\nc\cN{{\cal N}} \nc\fN[1]{N\br*{#1}} \nc\fn[1]{n\br*{#1}} \nc\rmN{\mathrm{N}} \nc\rmn{\mathrm{n}}
\nc\bfo{{\bf{o}}}\nc\bfO{{\boldsymbol O}}\nc\cO{{\cal O}} \nc\fO[1]{O\br*{#1}} \nc\fo[1]{o\br*{#1}} \nc\rmO{\mathrm{O}} \nc\rmo{\mathrm{o}}
\nc\bfp{{\bf{p}}}\nc\bfP{{\boldsymbol P}}\nc\cP{{\cal P}} \nc\fP[1]{P\br*{#1}} \nc\fp[1]{p\br*{#1}} \nc\rmP{\mathrm{P}} \nc\rmp{\mathrm{p}}
\nc\bfq{{\bf{q}}}\nc\bfQ{{\boldsymbol Q}}\nc\cQ{{\cal Q}} \nc\fQ[1]{Q\br*{#1}} \nc\fq[1]{q\br*{#1}} \nc\rmQ{\mathrm{Q}} \nc\rmq{\mathrm{q}}
\nc\bfr{{\bf{r}}}\nc\bfR{{\boldsymbol R}}\nc\cR{{\cal R}} \nc\fR[1]{R\br*{#1}} \nc\fr[1]{r\br*{#1}} \nc\rmR{\mathrm{R}} \nc\rmr{\mathrm{r}}
\nc\bfs{{\bf{s}}}\nc\bfS{{\boldsymbol S}}\nc\cS{{\cal S}} \nc\fS[1]{S\br*{#1}} \nc\fs[1]{s\br*{#1}} \nc\rmS{\mathrm{S}} \nc\rms{\mathrm{s}}
\nc\bft{{\bf{t}}}\nc\bfT{{\boldsymbol T}}\nc\cT{{\cal T}} \nc\fT[1]{T\br*{#1}} \nc\ft[1]{t\br*{#1}} \nc\rmT{\mathrm{T}} \nc\rmt{\mathrm{t}}
\nc\bfu{{\bf{u}}}\nc\bfU{{\boldsymbol U}}\nc\cU{{\cal U}} \nc\fU[1]{U\br*{#1}} \nc\fu[1]{u\br*{#1}} \nc\rmU{\mathrm{U}} \nc\rmu{\mathrm{u}}
\nc\bfv{{\bf{v}}}\nc\bfV{{\boldsymbol V}}\nc\cV{{\cal V}} \nc\fV[1]{V\br*{#1}} \nc\fv[1]{v\br*{#1}} \nc\rmV{\mathrm{V}} \nc\rmv{\mathrm{v}}
\nc\bfw{{\bf{w}}}\nc\bfW{{\boldsymbol W}}\nc\cW{{\cal W}} \nc\fW[1]{W\br*{#1}} \nc\fw[1]{w\br*{#1}} \nc\rmW{\mathrm{W}} \nc\rmw{\mathrm{w}}
\nc\bfx{{\bf{x}}}\nc\bfX{{\boldsymbol X}}\nc\cX{{\cal X}} \nc\fX[1]{X\br*{#1}} \nc\fx[1]{x\br*{#1}} \nc\rmX{\mathrm{X}} \nc\rmx{\mathrm{x}}
\nc\bfy{{\bf{y}}}\nc\bfY{{\boldsymbol Y}}\nc\cY{{\cal Y}} \nc\fY[1]{Y\br*{#1}} \nc\fy[1]{y\br*{#1}} \nc\rmY{\mathrm{Y}} \nc\rmy{\mathrm{y}}
\nc\bfz{{\bf{z}}}\nc\bfZ{{\boldsymbol Z}}\nc\cZ{{\cal Z}} \nc\fZ[1]{Z\br*{#1}} \nc\fz[1]{z\br*{#1}} \nc\rmZ{\mathrm{Z}} \nc\rmz{\mathrm{z}}
\DeclareMathOperator{\supp}{supp}
\DeclareMathOperator{\rank}{rank}

\DeclareMathOperator{\prob}{Pr}
\DeclareMathOperator{\bern}{Bernoulli}
\nc{\NormO}[1]{||#1||_{\ell_0}}
\nc{\Prob}[1]{\prob\br*{#1}}
\nc{\Exp}[1]{\exp\br*{#1}}
\nc{\Supp}[1]{\supp\br*{#1}}
\nc{\Log}[1]{\log\br*{#1}}
\nc{\Rank}[1]{\rank\br*{#1}}
\nc{\Min}[1]{\min\br*{#1}}
\nc\real{{\mathbb R}}
\nc\complex{{\mathbb C}}
\nc\F{{\mathbb F}}
\nc\N{{\mathbb N}}
\nc\bK{{\mathbb K}}
\nc\integers{{\mathbb Z}}
\nc\rationals{{\mathbb Q}}
\usepackage{multirow}
\usepackage{booktabs,caption}
\usepackage[flushleft]{threeparttable}
\usepackage{adjustbox}

\begin{document}
\title{Novel Impossibility Results for Group-Testing}
\author{Abhishek Agarwal \thanks{University of Minnesota, Twin Cities. Email: abhiag@umn.edu}\and
Sidharth Jaggi\thanks{Chinese University of Hong Kong. Email: jaggi@ie.cuhk.edu.hk} \and
Arya Mazumdar \thanks{University of Massachusetts, Amherst. Email: arya@cs.umass.edu}
}
\date{\vspace{-5ex}}
\maketitle

\begin{abstract}	
In this work we prove non-trivial impossibility results for perhaps the simplest non-linear estimation problem, that of {\it Group Testing} (GT), via the recently developed Madiman-Tetali inequalities. Group Testing concerns itself with identifying a hidden set of $d$ {defective} items from a set of $n$ items via $t$ {disjunctive/pooled} measurements (``group tests"). We consider the { linear sparsity regime}, i.e. $d = \delta n$  for any constant $\delta >0$, a hitherto little-explored (though natural) regime. In a standard information-theoretic setting, where the tests are required to be {non-adaptive} and a small probability of reconstruction error is allowed, our lower bounds on $t$ are the {\it first} that improve over the classical counting lower bound, $t/n \geq H(\delta)$, where $H(\cdot)$ is the binary entropy function. As corollaries of our result, we show that (i) for $\delta \gtrsim 0.347$, individual testing is essentially optimal, i.e., $t \geq n(1-o(1))$; and (ii) there is an {adaptivity gap}, since for $\delta \in (0.3471,0.3819)$  known {adaptive} GT algorithms require fewer than $n$ tests to reconstruct $\cD$, whereas our bounds imply that the best nonadaptive algorithm must essentially be individual testing of each element. Perhaps most importantly, our work provides a framework for combining combinatorial and information-theoretic methods for deriving non-trivial lower bounds for a variety of non-linear estimation problems.
% \cite{madimanTetali09}
\end{abstract}

\section{Introduction}

\begin{figure}
\centering
\begin{center}
\includegraphics[height=7.5cm]{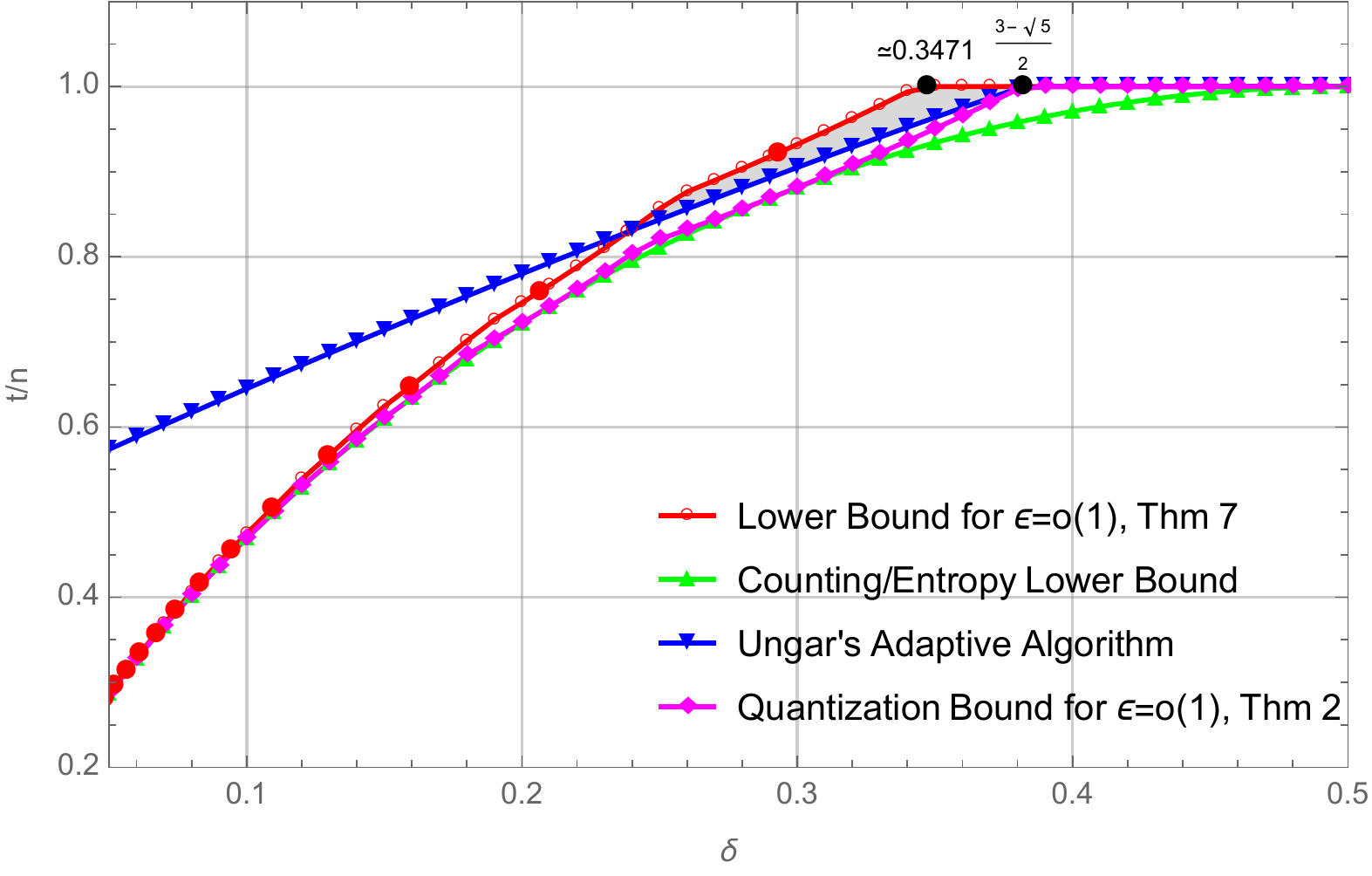} \\
\caption{Bounds on $t/n$ vs.  $\delta$ for $\epsilon = o(1)$.  The lower bound implied by \cref{thm:c_gt_1_2} corresponds to the horizontal part of the magenta curve,  and the result implied by \cref{thm:quantization_Gap} corresponds to the remainder of the magenta curve (the ``Quantization bound"). Both of these are superseded by the more sophisticated (and harder to prove) lower bound in \cref{thm:general_nagt}, plotted via the red curve. The shaded region (above the blue curve and below the red curve) denotes where there is an ``adaptivity gap'' -- the lower bound for (vanishing-error) NAGT exceeds the rate achievable by (zero-error) AGT~\cite{ungar1960cutoff}. }
\label{fig:plot}
\end{center}
\end{figure}

Estimation/inverse problems are the bread and butter of engineering -- given a system with a known input-output relationship, an observed output, and statistics on the input, the goal is to infer the input. While much is known about {\it linear} estimation problems and their fundamental limits \cite{kailath2000linear, reeves2013approximate}, understandably characterizing the fundamental limits of {\it non-linear} estimation problems are considerably more challenging.

Arguably one of the ``simplest" non-linear estimation problems is that of Group Testing (GT).  It is assumed that hidden among a set of $n$ items is a special set ${\cD}$ of $d$ {\it defective} items.\footnote{It is typically assumed that the value of $d$, or a good upper bound on it, is known {\it a priori}. This is because it can be shown that PAC-learning the value of $d$ is ``cheap" in terms of the number of group tests required \cite{damaschke2010bounds}.} The classical problem as posed by Dorfman \cite{dorfman1943detection}, requires one to exactly estimate ${\cD}$ via disjunctive measurements (``group tests") on ``pools" of items. That is, the output of each test is positive if the pool contains at least one  item from ${\cD}$, and negative otherwise. Besides its intrinsic appeal as a fundamental estimation problem, group-testing and its generalizations have a variety of diverse applications, such as bioinformatics \cite{ngo2000survey}, wireless communications \cite{wolf1985born,berger1984random}, and pattern finding \cite{macula2004group}.

Group testing problems come in a variety of flavours. In particular: 

\begin{enumerate}

\item {\bf (Non)-Adaptivity:} The testing algorithm can be {\it adaptive} (tests may be designed depending on previous test outcomes) or {\it non-adaptive} (tests must be designed non-adaptively, allowing for parallel testing/standardized hardware). 

\item {\bf Reconstruction error:} The reconstruction algorithm might need to be {\it zero-error} (always output the correct answer), or {\it vanishing error} (the probability of error goes to zero asymptotically in $n$), or an $\epsilon$ probability of  error ($\epsilon$-error) may be allowed. \footnote{Note that the error here is in the decoder, {\it not} in the test outcomes. There is considerable other literature (e.g.~\cite{chan2014non}) for the scenario when the test outcomes themselves may be noisy, for instance due to faulty hardware.}

\item {\bf Statistics of {$\cD$}:} Different works consider different statistical models for {$\cD$}. In {\it Combinatorial Group Testing} (CGT), it is assumed that any set of $d$ items may be defective, whereas in {\it Probabilistic Group Testing} (PGT), items are assumed to be i.i.d. defective with probability $d/n$. 

\item {\bf Sparsity regime:} Finally, it turns out that the specific {\it sparsity regime} matters - the regime where $d$ scales sub linearly in $n$ has seen much work, whereas the {\it linear sparsity regime} ($d=\delta n$ for some constant $\delta$) is relatively little explored.

\end{enumerate}

In this work we focus on non-adaptive group-testing with $\epsilon$-error in the linear sparsity regime -- indeed, this is perhaps the most ``natural" version of the problem, especially when viewed through an information-theoretic lens (for instance, the most investigated/used versions of channel codes are: non-adaptive since the encoder does not get to see the decoder's input; allow for reconstruction error; and typically have constant rate and hence are in the linear regime). Nonetheless, to put our own results in context we first briefly reprise the literature for other flavors of the problem in \cref{table:Results}. Note that, with a slight abuse of notation, we denote by $H(X)$ the entropy of the random variable/vector $X$, as well as the binary entropy function $H(x) = -x \log x -(1-x)\log (1-x)$. This should be clear from the argument of the function. 

In particular, let us briefly discuss the existing results of $\epsilon$-error nonadaptive group testing problem, the focus of this paper. It is quite straightforward to come up with a converse result based on counting/Fano's inequality (for example, see \cite{chan2014non}) that says $t \ge (1-\epsilon) \log \binom{n}{d}$. In \cite{aldridge2016improved}, it has been shown that this bound is also  tight for small $\epsilon$, as long as $d = O(n^{1/3})$ by showing randomized achievability schemes. Probabilistic existence of achievability schemes in this regime has also been derived, including for more general settings, in \cite{zhigljavsky2003probabilistic} (see Theorem~5.5 therein).  If we are allowed to sacrifice a constant factor in the number of tests, then we can have explicit deterministic construction of such achievability schemes \cite{mazumdar2016nonadaptive}. 
It is to be noted that, there is a surprising lack of study in the regime where the number of defectives varies linearly with the number of elements, i.e., $d = \delta n$. The counting converse bound simply boils down to 
$t\ge nH(\delta).$ This implies that individual testing of items is optimal when $\delta >0.5$. There is no other nontrivial converse bound that exists for the linear regime. In this paper we aim to close this gap. 
On the other hand, a recent work by Wadayama \cite{wadayama2017nonadaptive}, provides an achievability scheme in this regime based on sparse-graph codes (and density-evolution analysis). For certain values of $\delta$ (for example $\delta = 1-\frac1{2^{1/6}}$), this achievability scheme is in direct contradiction with our impossibility result in \cref{thm:general_nagt}. 
%\cite[Thm 2]{wadayama2017nonadaptive} implies that there exists a NAGT matrix with constant column weight $l=3$, and row weight $r=6$ that provides vanishing decoding error for $\delta = 1-\frac1{2^{1/6}}$. The relative number of tests for this matrix is $t/n = l/r = 1/2$. In contradiction, our impossibility scheme shows that $t/n \geq g_{\delta}^{-1}( H(\delta) )\vert_{\delta = 1-\frac1{2^{1/6}}} \geq .5038$

% On the other hand, a recent work by Wadayama \cite{wadayama2017nonadaptive}, provides achievability scheme in this regime based on sparse-graph codes (and density-evolution analysis). According to Wadayama's result the counting bound is asymptotically achievable for many points in the trade-off between $t/n$ and $\delta$. However this is in direct contrast with the converse results we obtained (as summarized in \cref{thm:general_nagt}) which are stronger than the counting bound.

It also is worth pointing out that the linear-sparsity regime is well-studied for adaptive group testing starting as early as in the sixties.\footnote{The authors would like to thank Matthew Aldridge for drawing our attention to this part of the literature.} It has been shown that under a zero probability of error metric, for $\delta > \frac{3-\sqrt{5}}{2}$ individual testing is the optimal strategy \cite{ungar1960cutoff}.~\footnote{In the literature there is also a conjecture~\cite{E} that if one demands that the worst-case number of tests (rather than average number of tests) be less than $n$, then under a zero probability of error metric no value of $\delta >1/3$ can be tolerated.} On the other hand, a rather simple adaptive algorithm achieves an expected number of tests equaling at most $0.5(3-(1-\delta)-(1-\delta)^2)n$ and identifies all defectives \cite{ungar1960cutoff}\footnote{~\cite{ungar1960cutoff} actually ascribes this algorithm to folklore -- we have been unable to find an earlier reference to this result.} -- we reprise this algorithm for completeness in \cref{sec:FKW_algo}. This is interesting if we contrast this with our converse result. There is a regime of values of $\delta$ (roughly in the range $\delta \in (0.3471,0.3819)$), where zero-error adaptive algorithms on average require fewer than $n$ tests to reconstruct ${\cal D}$, however our bounds imply that the best nonadaptive algorithm (even with vanishing error) turns out to essentially be individual testing of each element. 

%COMPLETE THE TABLE AND BRIEFLY TALK ABOUT THE RESULTS IN THE TABLE. +WADAYAMA  \cite{wadayama2017nonadaptive}

\begin{table}
\centering
\begin{adjustbox}{max width=\textwidth}
\begin{threeparttable}[c]
\begin{tabular}{|c|c|c|c|c|c|c|}
\hline
 & \multicolumn{2}{c|}{Adaptive} & \multicolumn{4}{c|}{Non-Adaptive} \\ \hline
 & \multicolumn{2}{c|}{Zero-error$^{ *1}$} & \multicolumn{2}{c|}{Zero-error} & \multicolumn{2}{c|}{Vanishing Error / $\epsilon$-error} \\ \hline
 & sub-linear  & \begin{tabular}{cc}
 	linear \\ $d = \delta n$
 \end{tabular} & \begin{tabular}{cc}
 	sub-linear \\ $d=n^\beta$
 \end{tabular} & \begin{tabular}{cc}
 	linear \\ $d = \delta n$
 \end{tabular} & \begin{tabular}{cc}
 	sub-linear \\ $d=n^\beta$
 \end{tabular}& \begin{tabular}{cc}
 	linear \\ $d = \delta n$
 \end{tabular} \\ \hline
 & &  &  &  &  & \\
Achievability & $d\log(n/d) +O(d)$ \cite{A} & $0.5(3-(1-\delta)-(1-\delta)^2)n$  \cite{ungar1960cutoff} & $d^2 \log n$ \cite{hwang1987non} & $n$ & $c(\beta)(1+\epsilon)(1-\beta)d\log(n)$, \cite{aldridge2016improved}$^{ *4}$ & See discussion of~\cite{wadayama2017nonadaptive} \\ 
 & &  &  &  &  & \\ \hline
 & &  &  &  &  & \\
Converse & $d\log(n/d)$ (folklore) & 
\begin{tabular}{c} 
$nH(\delta)$ \cite{chan2014non}$^{ *2}$,\\
$n-1$ if $\delta > \log_3(1.5) \approx 0.3691$, (Max. \# tests) \cite{riccio2000sharper} $^{ *3}$,\\
$n$ if $\delta > \frac{3-\sqrt{5}}{2} \approx 0.3819$, (Avg. \# tests)\cite{ungar1960cutoff} $^{ *3}$
\end{tabular} 
& \begin{tabular}{c}$\frac{d^2}{2\log d} \log n$ \cite{d1982bounds}\\ $n$ if $\beta\geq1/2$ \cite[Thm 7.2.9]{du2000combinatorial}\end{tabular}& $n$ \cite[Thm 7.2.9]{du2000combinatorial} & $(1-\epsilon)(1-\beta)d\log(n)$ \cite{chan2014non} & 
\begin{tabular}{ll} 
$n(1-\epsilon)$  & $\delta \gtrsim 0.3471$ \\
$g_{\delta}^{-1}(H(\delta)-\epsilon)$ & $\delta \lesssim 0.3471$ \\
(see \cref{thm:general_nagt})
\end{tabular}
\\ 
& &  &  &  &  & \\ \hline
\end{tabular}
    \begin{tablenotes}
      \small
      \item *1 Adaptive algorithms with reconstruction error have not really been considered much in the literature. Most proposed algorithms naturally result in zero-error, and the only known converses that are tighter than the counting bound intrinsically rely on the zero-error nature of the problem.
      \item *2 This bound holds even for $\epsilon$-error.
      \item *3 It is known \cite{ungar1960cutoff} that $\frac{3-\sqrt{5}}2 \approx 0.382$ is the correct cutoff point for adaptive PGT, whereas it is conjectured that the cutoff point for adaptive CGT is 1/3 \cite{E}
      \item *4 $c(\beta) \defeq 
      \begin{cases} 
		 = 1 & \beta<1/3 \\
		>1 & \beta \in (1/3,1)
   \end{cases}
   $
    \end{tablenotes}
\caption{A comparison of known inner and outer bounds on the number of tests required in a variety of group-testing settings. See also \cref{fig:plot}.}
\label{table:Results}
\end{threeparttable}
\end{adjustbox}
\end{table}

\subsection{Our Contributions and Techniques}

The canonical method (variously called the information-theoretic bound, or the counting bound) for proving impossibility results for group-testing problems via information-theoretic methods is quite robust to model perturbations: it works for adaptive and non-adaptive algorithms, zero-error and vanishing error reconstruction error criteria, PGT and CGT, and sub-linear and linear regimes. This method (see the Appendix in \cite{chan2014non} for an example) generally proceeds as follows: 

\begin{enumerate}

\item {\bf Entropy bound on input:} One first bounds the entropy $H(\bfX_{[n]})$ of the $n$-length binary vector $\bfX_{[n]}$ describing the status of the $n$ items (this means, the entry corresponding to an element in $\bfX_{[n]}$ is $1$ if and only if the element is defective): this quantity equals $\log\br{\binom{n}{d}}$ in the CGT case, and $nH(d/n)$ in the PGT case\footnote{One can see directly via Stirling's approximation that for large $n$ these two quantities are equal, up to lower-order terms.}; then

\item {\bf Information (in)equalities/Fano's inequality:} One uses standard information equalities, the data-processing inequality, the chain-rule, and Fano's inequality to argue that any group-testing scheme must satisfy the inequality $H(\bfY_{[t]}) \geq H(\bfX_{[n]})-n\epsilon$ (here $\bfY_{[t]}$ is a binary vector describing the set of $t$ test outcomes (that means an entry in $\bfY_{[t]}$ is $1$ if and only if the corresponding test result is positive), and $\epsilon$ is a lower bound on the probability of error of the group-testing scheme); and then

\item {\bf Independence bound.} Since $\bfY_{[t]}$ is a binary vector, one uses the independence bound to argue that $H(\bfY_{[t]}) \leq t$, and thereby obtains a lower bound on the required number of tests $t$, as a function of $\epsilon$, and $H(\bfX_{[n]})$. 

\end{enumerate}

Perhaps surprisingly, even for such a non-linear problem as group-testing, for a variety of group-testing flavors (such as non-adaptive GT with vanishing error when  $d={\cal O}(n^{1/3})$~\cite{aldridge2016improved}) such a straightforward approach results in an essentially tight lower bound on the number of tests required. The key contribution of our work is to provide a tightening of the method above for the regimes where it is not known to be tight.

While we believe our generalization technique is also fairly robust to various perturbations of the group-testing model, we focus in this work on the problem of $\epsilon$-error non-adaptive PGT\footnote{As noted in the Remark at the end of \cref{sec:general_GT_matrix}, almost all the techniques in this paper go through even for CGT -- we highlight the current technical bottleneck there as well.} in the linear sparsity regime. Possibly our key insight is that for this problem variant is that step (iii) of the counting bound may be quite loose.

Specifically, we present three novel converse bounds in \cref{thm:c_gt_1_2,thm:quantization_Gap,thm:general_nagt} for the general non-adaptive PGT problem in the linear regime. The result in \cref{thm:c_gt_1_2} follows from the observation that, for $\delta \geq \frac{3-\sqrt5}2$ the individual test entropies are maximized when each test contains exactly one object. Another simple result, for $\delta\leq \frac{3-\sqrt5}2$, in \cref{thm:quantization_Gap} follows from the observation that the individual test entropy, satisfies $ H(Y_l) \lneq 1$ for most of the region $\delta \in (0,1)$ because of the constraint that each test must contain an integer number of objects. 

Our main result (tighter than either \cref{thm:c_gt_1_2} or \cref{thm:quantization_Gap}, but also significantly more challenging to prove) in \cref{thm:general_nagt} exploits the observation that the tests in the  Non Adaptive Group Testing (NAGT) problem must have elements in common. For the linear regime, this observation leads to significant mutual information between the tests when the number of objects in the tests do not scale with $n$. Hence, we can exploit this mutual information to tighten the upper bound on the joint entropy $ H(\bfY_{[t]}) $ in step (iii) above. \Cref{fig:plot} plots our results in the linear regime along with existing results in the literature.

To bound the joint entropy $H(\bfY_{[t]})$ in step (iii), we must look for information inequalities that upper bound the joint entropies of correlated random variables. 
While the fascinating polymatroidal properties of such joint entropies ({\it Shannon-type inequalities}) explored  by Zhang and Yeung~\cite{zhang1998characterization}, as well as the {\it non-Shannon-type inequalities} that were subsequently found~\cite{zhang1997non} and are not consequences of such polymatroidal properties, are in this direction, they are perhaps {\it too} general to offer much guidance as to which specific information-inequalities might prove useful for providing non-trivial lower bounds for NAGT. 
A more structured characterization in this direction is Han's inequality~\cite{Han78} (implied by Shannon-type inequalities), that says
$$
H(\bfY_{[t]}) \le \frac{1}{t-1}\sum_{i=1}^t H(\bfY_{[t]\setminus\{i\}}),
$$
where $\bfY_{[t]\setminus\{i\}}$ contains test results except for the $i$th test. 

In this paper we use a significant generalization of Han's inequality to an {\it asymmetric} setting due to Madiman and Tetali  \cite{madimanTetali09}, that seems well-suited to analyzing the combinatorial structures naturally arising in NAGT.
Consider the NAGT matrix $M\in \set{0,1}^{t\times n}$, whose $(i,j)$th element is 1 if and only if the $i$th test includes the $j$th element. Let $\bfY_S$ denote the binary random variables corresponding to the test outcomes for $S\subseteq [t]$ and let $\bfX_{S}$ denote the indicator random variables corresponding to the objects for $S \subseteq [n]$. To demonstrate that non-trivial correlation between at least some sets of tests that must exist in our setting, we use the Madiman-Tetali inequalities \cite{madimanTetali09},
\begin{equation}\label{MT_conditional_form}
	H(\bfY_{[t]}) \leq \sum_{S \in \cC} \alpha(S) H(\bfY_{S} | \bfY_{S_-})
\end{equation}
where $S_- \defeq \set{i\in [t] : i \lneq j, \forall j\in S}$ and $\cC \subseteq 2^{[t]}$. The coefficients $\alpha(S)$ and the set $\cC$ form a cover of $2^{[t]}$ (more detail on this will be given in \cref{sec:impossibility_results}). In \cref{thm:general_nagt} we use the {\it weak form} \cref{MT} of the inequality above -- see \cref{sec:fut_work} for a discussion of the {\it strong form} and its potential use.

We use a two-step procedure to bound the joint entropy. In the first step, we assume that all the rows of the matrix $M$ has same weight (i.e., all tests contain the same number of elements, \cref{sec:constant_row_wt}). The results are then extended to general group testing matrices by considering them as a union of tests of (differing) constant weights. The final result is summarized in \cref{thm:general_nagt}.

%Madiman
%
%PLEASE GIVE A QUICK OVERVIEW OF the other INFORMATION INEQUALITIES ON JOINT ENTROPIES (INCLUDING ZHANG-YEUNG INEQUALITY), HAN?S INEQUALITY)

In the rest of the paper, we first describe our converse results \cref{sec:impossibility_results}, followed by a comparison with earlier bounds \cref{comparison} and  future directions of this project.

\section{Impossibility Results for Nonadaptive Group Testing} % (fold)
\label{sec:impossibility_results}
\subsection{Notation and Model}\label{sec:model}

For integers $a,b$ let $[a,b]\defeq \set{a,a+1,\ldots, b}$ and $[b]\defeq [1,b]$. Let $\log(.)$ denote the logarithm to the base 2, unless otherwise stated.

Consider the PGT problem with $n$ objects. Assume that we can tolerate a error probability $\epsilon$ in the decoding. Denote the indicator random variable which corresponds to object $i\in [n]$ being defective by $X_i$. Then $X_i$ are iid $\bern(\delta)$. With a slight abuse of notation, we use $X_i$ to refer to the random variable {\em and} the object $i$ interchangeably, when there is no scope of confusion.

Let $M \in \set{0,1}^{t\times n}$ denote the fixed GT matrix with $t$ tests. Denote the random variable corresponding to the outcome of the test in row $l$ by $Y_l$ and let $\bfY_{S} \defeq \set{Y_l}_{l\in S}$ for $S\subset [t]$. For an object set $R\subseteq [n]$, let $Y(R)$ denote the random variable corresponding to the test with object set $R$. For a class of object sets $\cR \subseteq 2^{[n]}$, let $\bfY(\cR)$ denote the random vector corresponding to the test with object sets $R \in \cR$.

Let $R_l\subseteq [n], l\in [t]$ denote the set of objects included in test $Y_l$ and let $S_i \subseteq [t], i\in [n]$ denote the tests containing the object $i$. Let $\cR(S), S\subseteq [t]$ denote the class of subsets of $[n]$ corresponding to the object sets of the tests $\set{Y_l}_{l\in S}$ ie. $\cR\br{S} \defeq \set{R_l}_{l\in S}$. For a class of sets $\cX \subseteq 2^\Omega$, and $A\in 2^\Omega$ define $\cX-A \defeq \set{Q\setminus A : Q\in \cX }$ as the class with set $A$ removed from all subsets in $\cX$ and let $\displaystyle [\cX] \defeq \bigcup_{A\in \cX} A$.

\subsection{Simple Converse Bounds}

Recall that in the linear sparsity regime each element is defective with probability $d/n = \delta$.
The canonical counting bound for the Group Testing problem gives the following upper bound on the number of tests for the $\epsilon$-error case:
\begin{equation}\label{counting_bound}
	t/n \geq H(\delta) - \epsilon
\end{equation}
This method uses the independence bound to get an upper bound on the joint entropy of the tests, \cref{independence_bound}, and then uses Fano's inequality, \cref{Fano}, to get a a lower bound on $t$.
\begin{equation}\label{independence_bound}
	H(\bfY_{[t]}) \leq t 
\end{equation}
\begin{equation}\label{Fano}
	\underbrace{H(\bfX_{[n]})}_{= n H(\delta)}  \leq H(Y_{[t]}) + \epsilon n.
\end{equation}

We tighten  \cref{counting_bound} by improving the bound in \cref{independence_bound} for the non-adaptive PGT problem in the linear regime. We do this by exploiting the fact the in the NAGT problem there would be a significant fraction of tests that have elements in common. Intuitively, we would want to maximize the entropy of the individual tests $\set{Y_l}_{l\in{[t]}}$ by choosing $\abs{R_l}$ such that $ H(Y_l) = H((1-\delta)^{\abs{R_l}}) = 1 $ i.e.  $ R_l \approx k_0(\delta)$ for $ l\in {[t]} $ where
\begin{equation}\label{k_0}
	k_0(\delta) \defeq \frac{\log\br{1/2}}{\log(1-\delta)}
\end{equation}
This implies that all tests contain a constant (with respect to $n$) number of objects. When any set $ S\subseteq [t] $ of such tests $ \bfY_S $ have an object in common, we can bound their joint entropy away from $ \abs{S} $. We exploit this fact to bound the joint entropy $ H(\bfY_{[t]}) $  away from $t$. But first, we exploit the nature of the group tests to improve \cref{counting_bound}. % for $\delta\geq 1/2$.
\begin{theorem}\label{thm:c_gt_1_2}
	For the PGT problem, we need at least $n (1-\epsilon/H(\delta))$ tests to identify the defective set with error probability $\epsilon$ for $\delta \geq \delta^\star$
	where $$\delta^\star \defeq \frac{3-\sqrt{5}}2$$.
\end{theorem}
\begin{proof}
Using the entropy chain rule, for $ \delta \geq \delta^\star$, we have,
\begin{subequations}	
\begin{align}
	H(\bfY_{[t]}) &\leq \sum_{l \in [t]} H(Y_{l}) \\
		&= \sum_{l\in [t]} H((1-\delta)^{\abs{R_l}}) \\
		&\leq t H(1-\delta) = t H(\delta)\label{delta_star}
\end{align}
\end{subequations}
Inequality~\ref{delta_star} is obvious for $\delta \geq 1/2$. For $\delta \in [\delta^\star,1/2)$, $(1-\delta)^2\leq \delta < 1/2 \implies H((1-\delta)^2) \leq H(\delta) = H(1-\delta)$. Hence, \cref{delta_star} follows here as well. Now, using \cref{Fano} and \cref{delta_star} we get,
\begin{equation*}
	t\geq n (1-\epsilon/H(\delta))
\end{equation*}
\end{proof}

Thus, for $\delta\geq \delta^\star$ we cannot do any better than individual testing. In the rest of the section, we focus on the GT bound for $\delta\leq \delta^\star$. Even in this regime, we can use the fact that \cref{k_0} is not an integer for all values of $\delta$ to improve \cref{counting_bound} without much effort.

\begin{theorem}\label{thm:quantization_Gap} [Quantization Bound]
	$$ t/n  \geq \frac{ H(\delta) - \epsilon }{\displaystyle \max_{k\in \N} H((1-\delta)^k)}$$
\end{theorem}
\begin{proof}
Due to the fact that each test can contain only an integer number of objects, we have 
\begin{align}\label{eq:quantization_bound}
	H(Y_l) = H((1-\delta)^{\abs{R_l}}) \leq \max_{k\in \N} H((1-\delta)^k) \nonumber\\
	\implies H(\bfY_{[t]}) \leq \sum_{l\in [t]} H(Y_l) \leq t \max_{k\in \N} H((1-\delta)^k)
\end{align}
Hence \cref{thm:quantization_Gap} follows from \cref{Fano} and \cref{eq:quantization_bound}.
\end{proof}

Note that, for $\delta \notin \set{1-\frac{1}{2^{1/k}}}_{k\in \N}$, $\max_{k\in \N} H((1-\delta)^k) \lneq 1$. Therefore, the result in \cref{thm:quantization_Gap} improves over the classical counting bound.

\subsection{Upper Bound via Madiman Tetali inequality}
To improve \cref{independence_bound} further for all values of $\delta\leq \delta^\star$, we use the Madiman Tetali inequalities in \cite{madimanTetali09} to exploit the correlation between tests,
\begin{equation}\label{MT_0}
	H(\bfY_{[t]}) \leq \sum_{S \in \cC} \alpha(S) H(\bfY_{S} )
\end{equation}
where $\cC$ are a class of subsets of $[t]$ that cover $[t]$, and $\set{\alpha(S)}_{S\in \cC}$ denote a fractional cover of the hypergraph $\cC$ on vertex set $[t]$. This means that for each $i \in [t]$, the set of numbers $\set{\alpha(S)}_{S\in \cC}$ satisfy the relation $\sum_{S \in \cC: i \in S} \alpha(S) \ge 1$.

Note that using the independence bound for $H(\bfY_{S} )$ in \cref{MT_0} we have,
\begin{equation}
	H(\bfY_{[t]}) \leq \sum_{S \in \cC} \alpha(S) \abs{S}
\end{equation}
where $\sum_{S \in \cC} \alpha(S) \abs{S} \geq t$. Therefore, to improve \cref{independence_bound} we have to utilize the fact that $\bfY_{S}$ have joint entropy less than $\abs{S}$. Heeding this intuition, first for a fixed set $S\subseteq [t]$, we derive a non-trivial upper bound on $ H(\bfY_S) $ in \cref{sec:ub_ent_Ys}, for tests $\bfY_S$ such that all of them have at least one object $ X \in \set{X_i}_{i\in [n]}$ in common ie $X \in \cap_{l\in S} R_l$. Next, we use this bound to derive a closed-form expression for the joint entropy $ H(\bfY_{[t]}) $ in \cref{MT_0} for a constant row weight  NAGT matrix $M$ in \cref{sec:constant_row_wt}. Finally, we generalize the upper bound to derive a closed form expression for arbitrary row weight matrices in \cref{sec:general_GT_matrix}. Using this expression and \cref{Fano}, we get an improvement over the counting lower bound in \cref{thm:general_nagt}.

\subsection{Upper bound on \texorpdfstring{$ H(\bfY_S) $}{H(Y\_S)}}\label{sec:ub_ent_Ys}

% \subsubsection{Proof of \texorpdfstring{\cref{closed_form_exp}}{eq} : Form of \texorpdfstring{$g_{\delta,k}(T)$}{g(T)}}\label{subsection:main_subsection}

Consider a set $S\subseteq [t]$ such that there exists an object $ X \in \set{X_i}_{i\in [n]}$ that is common in all the tests $\bfY_S$. Also assume that, $\abs{R_l} = k, \forall l\in S$. In this case, we upper bound the joint entropy of the tests $\bfY_S$ in \cref{thm:ub_ent_Ys}.
\begin{theorem}\label{thm:ub_ent_Ys}	
Consider $S\subseteq [t]$, such that $\abs{R_l} = k$, $\forall l\in S$ and all tests $\bfY_S$ have at least one object in common. Then,
\begin{equation}
	H(\bfY_S) \leq (1-\delta) \abs{S} H((1-\delta)^{k-1}) + H(\delta) - f_{\delta,k}(\abs{S})
\end{equation}
where
\begin{equation}\label{f_c_k}
	f_{\delta,k}(s) \defeq (\delta + (1-\delta) {p_{\delta,k}}^s) H\left (\frac{\delta}{\delta + (1-\delta) {p_{\delta,k}}^s} \right )
\end{equation}
and
\begin{equation}\label{p_c_k}
	p_{\delta,k} \defeq (1 - (1-\delta)^{k-1})
\end{equation}
\end{theorem}

In the rest of this section, we give the proof of \cref{thm:ub_ent_Ys}. Assume that the tests $\bfY_S$ have object $X\in \set{X_i}_{i\in [n]}$ in common. Let $ \bfY^\prime_{S}$ denote the set of tests containing the same objects as $\bfY_{S}$ but with object $X$ removed from all tests ie. $ \bfY^\prime_{S} \defeq \bfY(\cR(S)-\set{X})$. We have,
\begin{gather}
	H(\bfY_{S} ) = H(\bfY_{S} | X) + H(X )  - H(X | \bfY_{S} ) \label{acrobatics}\\
	H(\bfY_{S} | X) = (1-\delta) H(\bfY^\prime_{S}) \leq \abs{S} (1-\delta) H((1-\delta)^{k-1}) \label{approx1}\\
	H(X ) = H(\delta)  \label{approx2}\\
	H(X | \bfY_{S} ) =  (\delta + (1-\delta) \Prob{\bfY^\prime_{S} = {\bf{1}}}) H \left (\frac{\delta}{\delta + (1-\delta) \Prob{\bfY^\prime_{S} = {\bf{1}}}} \right ) \label{approx3}
\end{gather}

Therefore, combining \cref{acrobatics}, \cref{approx1}, \cref{approx2}, and \cref{approx3}, we have
\begin{align}\label{approx}
	H(\bfY_S ) 
		&\leq  (1-\delta) H((1-\delta)^{k-1}) + H(\delta) -  H(X | \bfY_{S} ) 
\end{align}

Note that, 
\begin{equation}\label{positive_diff}
	\frac{\partial (\delta + (1-\delta) x) H(\frac{\delta}{\delta + (1-\delta) x})}{\partial x} = (1-\delta) \log \br*{1 + \frac{\delta}{x (1-\delta)}} \geq 0
\end{equation}
Thus, the expression for $H(X | \bfY_{S})$ in \cref{approx3} is minimized at the minimum possible value of $\Prob{\bfY^\prime_{S} = {\bf{1}}}$.  We lower bound the probability $\Prob{\bfY^\prime_{S} = \bf{1}}$ using \cref{lem:probY_S} to get an upper bound on \cref{approx}.
\begin{lemma}
\label{lem:probY_S} 
For any $S\subseteq [t]$,  we have,
\begin{equation}
	\min_{\substack{R_l \in \cR(S):\\ \abs{R_l} = r_l}}\Prob{\bfY_S={\bf1}} = \prod_{l\in S}(1-(1-\delta)^{r_l})
\end{equation}
\end{lemma}
\begin{proof}
Note that $\Prob{Y_l = 1} = (1-\delta)^{\abs{R_l}}$. We show that the minimization in \cref{eq:probY_S} occurs when all object sets $\set{R_l}_{l\in S}$ are disjoint. Since, in that case the tests in $\bfY_S$ are independent, we must have,
\begin{align}\label{eq:probY_S_1}
\min_{\substack{R_l \in \cR(S):\\ \abs{R_l} = r_l}}\Prob{\bfY_S={\bf1}} &= \prod_{l\in S} \Prob{Y_l={1}} \\
	&= \prod_{l\in S}(1-(1-\delta)^{r_l}) \nonumber
\end{align}

Without loss of generality, let $S=[s]$. Suppose that, the tests $\bfY_S$ are such that there exists an object $i$ that is common among tests $Y_1, Y_2,\ldots, Y_a$ for some $a\in [2,\abs{S}]$. Then, we show that, we can decrease the probability $\Prob{\bfY_{S} = \bf{1}}$ by modifying $R_1$ to $R_1^\ast$ by including an object  ${i^\ast} \in [n] \setminus [\cR(S)]$ in $Y_1$ instead of object $i$ such that $R_1^\ast = (R_1\setminus \set{i}) \cup \set{i^\ast}$. Denote the modified tests by $\bfY_S^{\ast}$. Then, it suffices to prove that,
\begin{equation}\label{probY_S}
	\Prob{\bfY_{S}= \bf{1}} \geq \Prob{\bfY^{\ast}_{S}= \bf{1}}
\end{equation}
since using \cref{probY_S} recursively for objects contained in more than one tests in $\bfY_S^\ast$ we can prove \cref{eq:probY_S_1}. We prove \cref{probY_S} in \cref{sec:proof_lemma_probY_s}.
\end{proof}

Thus, from \cref{lem:probY_S} we have,
\begin{equation}\label{eq:probY_S}
	\Prob{\bfY^\prime_{S} = \bf{1}} \geq p_{\delta,k}^{\abs{S}} ,\;\forall i\in [n]
\end{equation}
where $p_{\delta,k}$ is defined in \cref{p_c_k}. Hence, from \cref{approx3}, \cref{positive_diff}, and~\cref{eq:probY_S}, we have,
\begin{equation}\label{condEntLB}
	H(X|\bfY_S) \geq f_{\delta,k}(\abs{S})
\end{equation}
Now, combining \cref{approx,condEntLB} we have,
$$ H(\bfY_S) \leq  (1-\delta) \abs{S} H((1-\delta)^{k-1}) + H(\delta) -  f_{\delta,k}(\abs{S}) $$
where $f_{\delta,k}(s)$ is as defined in \cref{f_c_k}.

\subsection{Constant Row Weight Testing Matrix}\label{sec:constant_row_wt}

In this section we assume that matrix $M$ has constant row weight $k$ such that $k\geq 2$. Intuitively, this is  a very natural assumption. Since it allows each test in matrix $M$ to be symmetric. This assumption also allows us to easily upper bound the joint entropy of the tests $H(\bfY_{[t]})$ using \cref{MT_0}, as we see below.

To apply \cref{MT_0}, we consider the hypergraph $\cC$ with $n$ edges and having matrix $M$ as the incidence matrix. Thus, $\cC\defeq \set{S_i}_{i=1}^n$, where $S_i$ denotes the support set of column $i$ in $M$. Note that, in this case, $\alpha(S_i) = \frac{1}{k}$ forms a cover of the hypergraph $\cC$. Therefore, we have,
\begin{equation}\label{MT}
	H(\bfY_{[t]}) \leq \frac1k \sum_{i=1}^n H(\bfY_{S_i})
\end{equation}
We upper bound the expression on the RHS in \cref{MT} to get an asymptotic closed form expression for the joint entropy of the form,
\begin{equation}\label{closed_form_exp}
\frac{H(\bfY_{[t]})}n \leq  g_{\delta,k}(t/n)	
\end{equation}
where $g_{\delta,k}(T)$ is shown to be an increasing function of $T$. Thus, using \cref{Fano}, we have,
\begin{theorem}	\label{thm:constant_row_wt}
Consider the non-adaptive PGT problem, with tolerable probability of error $\epsilon$. Assume that each object is defective independently with probability  $\delta$. Then, for a constant row weight $k$ group testing matrix, we have asymptotically in $n$,
\begin{equation}
	\frac{t}n \geq g_{\delta,k}^{-1}( H(\delta) - \epsilon)
\end{equation}
where $g_{\delta,k}^{-1}(x) \defeq y$ such that $g_{\delta,k}(y) = x$ and 
\begin{equation}\label{g_c_k}
	g_{\delta,k}(T) \defeq T (1-\delta) H((1-\delta)^{k-1})  + \frac{1}{k} \br*{ H(\delta) - f_{\delta,k}(k \;T) } 
\end{equation}
where $f_{\delta,k}(T)$ is defined in \cref{f_c_k}.
\end{theorem}

The proof of \cref{thm:constant_row_wt} follows from \cref{closed_form_exp} and~\cref{Fano}. The form of $g_{\delta,k}(T)$ in \cref{closed_form_exp} is derived below as
\begin{align}
	H(\bfY_{[t]}) &\leq  t (1-\delta) H((1-\delta)^{k-1}) + \frac{n}{k} H(\delta) -  \frac{1}{k} \sum_{i=1}^{n}  f_{\delta,k}(\abs{S})  \label{MT_applied3} \\
		&\leq  t (1-\delta) H((1-\delta)^{k-1}) + \frac{n}{k} \br*{ H(\delta) -  f_{\delta,k}\left (\frac{1}{n} \sum_{i=1}^{n} \abs{S} \right ) }  \label{MT_applied4} \\
		&=  t (1-\delta) H((1-\delta)^{k-1}) + \frac{n}{k} \br*{ H(\delta) -  f_{\delta,k}(k t/n) }  \label{MT_applied5} \\
		&= n g_{\delta,k}(t/n) \nonumber
\end{align}
where \cref{MT_applied3} follows from \cref{thm:ub_ent_Ys} and \cref{MT}, and \cref{MT_applied4} follows from the convexity of $f_{\delta,k}(s)$ from \cref{lem:JensenApplied}. Note that since $f_{\delta,k}(s)$ is a convex decreasing function of $s$,  $g_{\delta,k}(T)$ must be a concave increasing function of $T$. Thus, \cref{closed_form_exp} and hence \cref{thm:constant_row_wt} follows.

{\lemma \label{lem:JensenApplied} $ f_{\delta,k}(s) \defeq (\delta + (1-\delta) {p_{\delta,k}}^s) H(\frac{\delta}{\delta + (1-\delta) {p_{\delta,k}}^s})$ is a convex decreasing function of $s$
}
\begin{proof}
	We have,
	\begin{equation}\label{derivative_fck}
		\frac{\partial f_{\delta,k}(s)}{\partial s} = (1-\delta) p_{\delta,k}^s \ln p_{\delta,k}  \log\br*{ 1+\frac{c}{(1-c)p_{\delta,k}^s} } <0
	\end{equation}
	and
	\begin{equation}\label{convexity_fck}
		\frac{\partial^2 f_{\delta,k}(s)}{\partial s^2} = \frac{(\ln p_{\delta,k})^2}{\ln 2} y \br*{ \underbrace{\ln (1+\delta/y) - \frac{\delta/ y}{1+\delta/y}}_{\Psi(\delta/y)} } > 0
	\end{equation}
	where $y = (1-\delta) p_{\delta,k}^s $. Since $\Psi(z) \defeq \ln(1+z)-\frac{z}{1+z}$ is always positive for $z>0$, $f_{\delta,k}(s)$ is a decreasing convex function of $s$ from \cref{convexity_fck,derivative_fck}.
\end{proof}

\subsection{General Testing Matrix}\label{sec:general_GT_matrix}

In this section, we remove the assumption that matrix $M$ has a fixed row weight $k$ and derive an upper bound on $ H(\bfY_{[t]}) $ -- better than \cref{independence_bound} -- for the most general case. We use this upper bound to improve \cref{counting_bound} in \cref{thm:general_nagt}. 

We separate the matrix $M$ into submatrices $\set{M_k}_k \in\set{0,1}^{t_k\times n}$ based on the number of objects in the tests. Thus, matrix $M_k$ has $t_k = \alpha_k t$ tests of weight $k$ such that $\sum_{k=1}^\infty \alpha_k = 1$. 

Now, we show that the analysis in \cref{sec:constant_row_wt} follows through for each matrix $M_k$. Let $t_0\defeq 0$. Assume w.l.o.g. that the tests corresponding to $M_k$ are $\bfY_{[t_{k-1}+1,t_k]}$. Denote the support sets of column $i$ in $M_k$ by $S_{k,i}$. Note that some of the columns in the matrix may be empty, i.e. $\abs{S_{k,i}} = 0$. Thus let $\cC_k^\prime$ denote the support sets corresponding to the non-empty columns. Let $\cC_k$ denote the class of support sets $\set{S_{k,i}}_{i\in [n]}$ where each empty column is considered a distinct set. Therefore we have,
\begin{align}\label{Mk_MT}
	H(\bfY_{[t_{k-1}+1,t_k]}) &\leq \sum_{S_{k,i} \in \cC_k^\prime} \alpha(S) H(\bfY_{S} ) \nonumber \\
		&= \frac{1}{k} \sum_{S_{k,i} \in \cC_k^\prime} H(\bfY_{S_{k,i}} | X_i) + H(X_{i} )  - H(X_{i} | \bfY_{S_{k,i}} )
\end{align}
When $\abs{S_{k,i}} = 0$, we have $ H(X_{i} | \bfY_{S_{k,i}} ) = H(X_{i})$ and $H(\bfY_{S_{k,i}} | X_i)  = 0$. Note that for $\abs{S_{k,i}}=0$ the lower bound in \cref{condEntLB} also gives $ H(X_{i} | \bfY_{S_{k,i}} ) \leq f_{\delta,k}(0)= H(\delta)$. Therefore, 
\begin{equation}\label{acrobatics_1}
	  \sum_{S_{k,i} \in \cC_k\setminus \cC_k^\prime} \br*{ H(\bfY_{S_{k,i}} | X_i) + H(X_{i} )  - H(X_{i} | \bfY_{S_{k,i}} ) } = 0 
\end{equation}
Hence, combining \cref{Mk_MT,acrobatics_1} we have
\begin{align}\label{Mk_MT_1}
	H(\bfY_{[t_{k-1}+1,t_k]}) &\leq \frac{1}{k} \sum_{S_{k,i} \in \cC_k} \br*{ H(\bfY_{S_{k,i}} | X_i) + H(X_{i} )  - H(X_{i} | \bfY_{S_{k,i}} ) }
\end{align}
\begin{remark}
Note that the manipulation in \cref{acrobatics_1}, although seemingly unnecessary, is required because $\abs{S_i} = 0$ is not possible in \cref{sec:constant_row_wt}. But since this is possible in this section with non-constant weight GT matrices, the lower bound of $H(X_{i} | \bfY_{S_{k,i}} ) $ in \cref{condEntLB} may not hold in this case. But this algebraic manipulation resolves that problem.
\end{remark}

Using the expressions in \cref{approx1}, \cref{approx2} and \cref{condEntLB} in \cref{Mk_MT_1}, we have,
\begin{align}\label{Mk_MT_2}
	H(\bfY_{[t_{k-1}+1,t_k]}) &\leq \frac{1}{k} \sum_{S_{k,i} \in \cC_k} \br*{ H(\bfY_{S_{k,i}} | X_i) + H(X_{i} )  - H(X_{i} | \bfY_{S_{k,i}} ) } \nonumber \\
		&\leq t_k (1-\delta) H((1-\delta)^{k-1})  + \frac{n}{k} \br*{ H(\delta) - f_{\delta,k}\left (\frac{1}{n} \sum_{i\in [n]}\abs{S_{k,i}} \right ) }   \nonumber\\ 
	\iff\frac{H(\bfY_{[t_{k-1}+1,t_k]})}{n} &\leq (t_k/n) (1-\delta) H((1-\delta)^{k-1})  + \frac{1}{k} \br*{ H(\delta) - f_{\delta,k}\left (\frac{k \;t_k}{n} \right ) }   \nonumber\\ 
		&= \alpha_k \frac{t}{n} (1-\delta) H((1-\delta)^{k-1})  + \frac{1}{k} \br*{ H(\delta) - f_{\delta,k}\left (k \alpha_k \frac{t}{n}\right ) }   \nonumber\\ 
		&= g_{\delta,k} \left ( \frac{\alpha_k t}{n} \right )
\end{align}
where for $k=1$, we define,
\begin{equation}\label{gk1}
	g_{\delta,1}(T) \defeq T H(\delta).
\end{equation}

Thus, we have from \cref{Mk_MT_2},
\begin{align}
	\frac{H(\bfY_{[t]})}{n} &\leq \sum_{k\geq 1} \frac{H(\bfY_{[t_{k-1}+1,t_k]})}{n} \\
		&\leq \sum_{k \geq 1}  g_{\delta,k}(\alpha_k t/n) \\
		&\leq  \max_{\substack{\set{\alpha_k}_k \\ : \sum_k \alpha_k = 1}}  \sum_{k \geq 1}  g_{\delta,k}(\alpha_k t/n) \label{Entropy_UB1}\\
		&=  \max_{k\geq 1} g_{\delta,k}(t/n) \label{Entropy_UB2}
		% &\leq t \max_{\substack{\set{\alpha_i}_i \\ : \sum_i \alpha_i = 1}} \sum_{i\geq 2} \alpha_i \br*{n H(\delta) \frac{1}{i} + (1-\delta) t H((1-\delta)^{i-1})}  \nonumber \\
		% &\leq t \; \max_i \frac{ H(\delta) + (1-\delta) H((1-\delta)^{i-1})  }{i} 
\end{align}
where~\cref{Entropy_UB2} follows since the maximization in~\cref{Entropy_UB1} is over a convex polytope and $\sum_{k \geq 1}  g_{\delta,k}(\alpha_k T) $ is a concave increasing function of $\set{\alpha_k}_{k\geq 1}$ from \cref{convexity_fck} and~\cref{gk1} and the following equations,
\begin{subequations}
	\begin{equation}
		\frac{\partial^2 \sum_{k\geq 1} g_{\delta,k}(\alpha_k T)}{\partial \alpha_{m}^2} =  - m \;T^2\; \frac{\partial^2 f_{\delta,m}(x)}{\partial x^2}\big\vert_{x = m \;\alpha_mT}
	\end{equation}
	\begin{equation}
		\frac{\partial^2 \sum_{k\geq 1} g_{\delta,k}(\alpha_k T)}{\partial \alpha_{m} \partial \alpha_{m^\prime}} =  0.
	\end{equation}
\end{subequations}

Then, from \cref{Entropy_UB2} and~\cref{Fano}, we have our main result.
\begin{theorem}\label{thm:general_nagt}
Consider the non-adaptive PGT problem with probability of error at most $\epsilon$. Assume that each object is defective independently with probability  $\delta$. Then, we have asymptotically in $n$,
\begin{equation}	\label{eq:finalUB}
	t/n \geq  g_{\delta}^{-1}( H(\delta) - \epsilon)
\end{equation}
where
\begin{equation}\label{eq:finalUB_func}
	g_{\delta}(T) \defeq \min_{k\in \N} g_{\delta,k}(T)
\end{equation}
\end{theorem}

The bound in \cref{thm:general_nagt} intersects with $t/n=1$ at $\delta \approx 0.3471$. 

{\bf Remark:} Note that although we have stated the results in this paper for the PGT problem, most arguments in the paper go through for the corresponding CGT problem as well. The only problem arises in the proof of \cref{lem:probY_S} since when $\abs{S}$ is not constant (w.r.t. $n$) $\Prob{\bfY_S = {\bf0}} \ne (1-\delta)^{\abs{[\cR(S)]}}$. However we believe that with some effort and appropriate approximations, our techniques should also go through for CGT.

\section{Discussion and Comparison}
\label{comparison}

In this section we compare the results in \cref{thm:general_nagt} with other achievability and impossibility results in the literature. First, to show an adaptivity gap, we consider a simple adaptive algorithm for the GT problem presented in \cite{C} and analyze the expected number of tests required. The algorithm is defined in \cref{sec:FKW_algo}. The expected number of tests performed is
\begin{equation}\label{adaptive_tests}
	 n \min \left \{1,\frac{1}2 \br*{3- (1-\delta) - (1-\delta)^2} \right \}.
\end{equation}

The graph in \cref{fig:plot} plots the lower bound in \cref{thm:general_nagt}, the expected number of tests in~\cref{adaptive_tests}, the quantization bound in \cref{thm:quantization_Gap}, and the entropy counting bound, ~\cref{counting_bound} for vanishing error i.e. $\epsilon = o(1)$. The solid circle markers in the plot represent the bound in~\cref{eq:finalUB} for $\delta$ such that $\frac{\log(1/2)}{\log\br{1-\delta}} \in \N$. From \cref{fig:plot}, there exists a non-vanishing gap between the lower bound in \cref{thm:general_nagt} and the counting bound. The quantization bound in \cref{thm:quantization_Gap} also improves over the counting bound for a significant region of $\delta$. As claimed earlier, we can also see an adaptivity gap in \cref{fig:plot} represented by the shaded region.

Even when the results in \cref{thm:general_nagt} are plotted for $\epsilon = o(1)$, we can see from \cref{eq:finalUB} and \cref{fig:plot} that there would exist a non-vanishing gap between \cref{eq:finalUB} and the counting bound for small values of $\epsilon$ as well. For $\epsilon>\delta$, it would be possible to ignore certain objects altogether during tests, and hence a smaller number of tests could be possible. 

The number of objects in each test in the GT matrix is constrained to be an integer. This gives a discrete nature to the bound in \cref{eq:finalUB}. This is evident from the piecewise nature of plot for the lower bound in~\cref{eq:finalUB}.

% section impossibility_results (end)

\section{Future Work / Implications}\label{sec:fut_work}

In this work we use the weak form of the Madiman-Tetali inequalities in \cite{madimanTetali09} to upper bound the joint entropy of the test $\bfY_{[t]}$. Since the weak form of the inequalities ignores the gains the conditional form of the entropy function provides, we suspect that there is a lot more to be gained by exploiting the strong form in \cref{MT_conditional_form}. Motivated by the results in this work, we conjecture that for {\em any} constant $\delta > 0$, $n - o(n)$ non-adaptive tests are {\em necessary} to ensure vanishing error.

%gains due at using the conditional MT inequality
From the the plots in \cref{fig:plot} and \cref{thm:general_nagt} we see that the joint entropy of the tests is minimized for row weight $k_0(\delta) = \frac{\log(1/2)}{\log(1-\delta)}$. As $\delta$ decreases (and $k_0(\delta)$ increases) the improvement in the first term ($H(\bfY_S | X_i)$) in \cref{acrobatics} reduces. For \cref{MT_conditional_form}, the strong form of the Madiman-Tetali inequalities, this term becomes
\begin{equation}\label{term1_MT}
H(\bfY_S | X_i, \bfY_{S_-}).
\end{equation}
Recall the definition of $S_-$ from \cref{MT_conditional_form}.
Intuitively, as $k_0(\delta)$ increases, the average mutual information between tests $\bfY_S$ and $\bfY_{S_-}$ increases. Thus, for the conditional Madiman-Tetali form, the term in \cref{term1_MT} may be a lot smaller for small $\delta$. Hence we believe that the bound in \cref{thm:general_nagt} could potentially be improved significantly by using the conditional form of the Madiman-Tetali inequalities. However, the analytical approximations involved in using these techniques are also non-trivial.

%gains due at changing the hypergraph $\cC$
Another way to see that the bound in \cref{thm:general_nagt} is loose is by changing the hypergraph $\cC$ in \cref{MT}. Instead of taking the support of a single column of $M$ as hyperedges in $\cC$ in \cref{MT}, we could use the union of support of $j$ columns, for $j>1$ i.e. $\cC = \displaystyle \left \{ \bigcup_{i\in A} S_i \right \}_{A\in \binom{[n]}{j} }$. For large $k_0(\delta)$ we can see that a large number of tests $\set{\bfY_S}_{S\in \cC}$ corresponding to the hyperedges $S\in \cC$ will have more than one object in common. Therefore, we believe there is still room for improvement even  just employing  the weak degree form of the Madiman-Tetali inequalities.

One more potentially promising direction worth exploring is to consider the rows or columns of the NAGT matrix as codewords of a binary code, use the {\it combinatorial} Delsarte inequalities~\cite{delsarte1973algebraic} that provide non-trivial bounds on the distance spectrum of codes to appropriately ``tighten" the {\it information-theoretic} Shannon-type inequalities (specifically the Madiman-Tetali inequalities) in this work. We are motivated by the fact that such an optimization approach has had significant success in providing the essentially tightest known upper bounds on the sizes of binary error-correcting codes~\cite{mceliece1977new} -- while we freely admit that it is unclear to us what such a fusion of combinatorial and information-theoretic techniques might look like concretely, nonetheless, the prospect is intriguing.

Finally we believe that our technique of lower bounding the number of tests via the Madiman-Tetali inequalities may have wide applicability in similar sparse recovery problems and other variants of group testing, such as threshold group testing~\cite{damaschke2006threshold}, the pooled-data problem~\cite{wang2016data}, and potentially even long-standing open problems pertaining to threshold secret-sharing schemes~\cite{beimel2011secret}.

%(ii) Combining with other combinatorial bounds such as Delsarte inequalities

%(iii) Potential extensions to adaptive GT, partial recovery, pooling tests, etc?

%GIVE DETAILS OF ABOVE.

\section{Acknowledgements} The authors would also like to thank Oliver Johnson, Matthew Aldridge, and Jonathan Scarlett for enlightening discussions that significantly improved this work. In particular we would like to acknowledge Oliver Johnson for drawing our attention to~\cite{ungar1960cutoff, riccio2000sharper}, to Matthew Aldridge for the observation that our lower bounds imply an adaptivity gap, and to all three for helpful discussions regarding~\cite{wadayama2017nonadaptive}.

This work was partially funded by a grant from the University Grants  Committee of the Hong Kong Special Administrative Region (Project No.\  AoE/E-02/08), RGC GRF grants 14208315 and 14313116 and NSF awards CCF 1642550 and CCF 1618512.

% !TEX root = Doc1.isit.v1.tex
\appendix

\section{Appendix}

\subsection{Proof of the remainder of \texorpdfstring{Lemma~\ref{lem:probY_S}}{Lemma} }\label{sec:proof_lemma_probY_s}

$$	\Prob{\bfY_{S}= \bf{1}} \geq \Prob{\bfY^{\ast}_{S}= \bf{1}}	$$
{\pf

Let $[\cR(\emptyset)] \defeq \emptyset$ and let $\rmE(R)$ denote the event when the test $Y(R)$ is negative, and $\rmE_l$ denote the event $\rmE(R_l)$. Let $\overline{\rmE}$ denote the complement of the event $\rmE$. Let $\rmE_{S} \defeq \bigcup_{l\in S} \rmE_l$. 

\step{1}{
Using inclusion-exclusion \cite[Section 2.1]{stanley2012numerative}, we have,
\begin{align}\label{eq:inc_exc}
	\Prob{\bfY_S = \bf{1}} &= \Prob{\bigcap_{l\in S} \overline{\rmE_l}} = 1 - \Prob{\rmE_S} \nonumber \\
		&= 1 - \sum_{U \subseteq {S} : U\ne \emptyset}  \Prob{\bigcap_{l\in U} \rmE_l} (-1)^{\abs{U}-1}\nonumber \\
		&= \sum_{U \subseteq S}  \Prob{\bigcap_{l\in U} \rmE_l}  (-1)^{\abs{U}} 
\end{align}
}

\step{2}{
Let $\rmE(R)$ denote the event when the test $Y(R)$ is negative, and $\rmE_l$ denote the event $\rmE(R_l)$. Let $\overline{\rmE}$ denote the complement of the event $\rmE$. Let $\rmE_{S} \defeq \bigcup_{l\in S} \rmE_l$. Then, for any $a \in [s]$,
\begin{equation}\label{eq:sublemma}
	\Prob{Y_1 = 0, \bfY_{[2,a]} \ne {\bf1},\bfY_{[a+1,s]} = {\bf1} } = \sum_{\substack{ U \subseteq [2,s] : \\ U\cap [2,a] \ne \emptyset}} \Prob{\bigcap_{l\in U}\rmE_l \cap \rmE_1} (-1)^{\abs{U}+1} 
\end{equation}

\begin{proof}
Using \stepref{1}, we have,
\begin{align}\label{eq:inc_exc_1}
	&\Prob{{\rmE} \cap \overline{\rmE_{[{a+1},s]}} } \nonumber\\
	&= \Prob{{\rmE}} - \Prob{{\rmE} \cap \br{\rmE_{[{a+1},s]} } } \nonumber\\
	&= \sum_{V \subseteq [a+1,s] }  \Prob{\rmE \cap \bigcap_{l\in V} \rmE_l} (-1)^{\abs{V}}
\end{align}

Again using \stepref{1} we have,
\begin{align}
	&\Prob{\bfY_{[a]}\ne {\bf1}, \bfY_{[a+1,s]} = {\bf1} } = \Prob{{\rmE_{[a]}} \cap \overline{\rmE_{[{a+1},s]}} } \nonumber\\
	&= \sum_{U \subseteq [a]: U\ne \emptyset} (-1)^{\abs{U}-1} \Prob{{\bigcap_{l\in U}\rmE_{l}} \cap \overline{\rmE_{[{a+1},s]}} } \nonumber\\
	&= \sum_{\substack{U \subseteq [a] : \\U\ne \emptyset}} \sum_{V \subseteq [a+1,s] }  \Prob{\bigcap_{l\in U}\rmE_{l} \cap \bigcap_{l\in V} \rmE_l} (-1)^{\abs{U\cup V}-1} \label{eq:inc_exc_2_1}\\
	&= \sum_{\substack{U \subseteq [s] :\\ U\cap [a]\ne \emptyset}} \Prob{\bigcap_{l\in U}\rmE_{l} } (-1)^{\abs{U}-1} \label{eq:inc_exc_2}
\end{align}
where \eqref{eq:inc_exc_2_1} follows from \eqref{eq:inc_exc_1}. Now \eqref{eq:sublemma} directly follows from \eqref{eq:inc_exc_2}
\end{proof}
}

\step{3}{We have, using \stepref{1} and \stepref{2},
\begin{align}
	& \Prob{\bfY_{S}= \bf{1}}  \nonumber\\
	&= \sum_{U \subseteq S }  (-1)^{\abs{U}} \Prob{ \bigcap_{i\in U }\rmE_l } \nonumber\\\
		&=  \sum_{\substack{ U \subseteq S :  \\  1\notin U \mbox{ or }\\ [2,a]\cap U = \emptyset} }  (-1)^{\abs{U}} \Prob{ \bigcap_{l\in U}\rmE_l } +   \underbrace{\sum_{\substack{ U \subseteq S\setminus \set{1} : \\ [2,a]\cap U \ne \emptyset}} (-1)^{\abs{U}+1} \Prob{ \bigcap_{l\in U}\rmE_l \cap \rmE_1} }_{ = \Prob{Y_1 = 0, \bfY_{[2,a]} \ne {\bf1},\bfY_{[a+1,s]} = {\bf1} }}  \nonumber\\\
		&\geq  \sum_{\substack{ U \subseteq S :  \\  1\notin U \mbox{ or }\\ [2,a]\cap U = \emptyset} }  (-1)^{\abs{U}} \Prob{ \bigcap_{l\in U}\rmE_l } + (1-\delta)  \sum_{\substack{ U \subseteq S\setminus \set{1} : \\ [2,a]\cap U \ne \emptyset}} (-1)^{\abs{U}+1} \Prob{ \bigcap_{l\in U}\rmE_l \cap \rmE_1}   \nonumber\\
		&=  \sum_{\substack{ U \subseteq S :  \\  1\notin U \mbox{ or }\\ [2,a]\cap U = \emptyset} }  (-1)^{\abs{U}} (1-\delta)^{\abs{ [\cR(U)] }}  + \sum_{\substack{ U \subseteq S\setminus \set{1} : \\ [2,a]\cap U \ne \emptyset}} (-1)^{\abs{U}+1} (1-\delta)^{\abs{ [\cR(U)] \cap R_1}+1}   \label{probY_S_ineq}\\
		&= \Prob{\bfY^{\ast}_{S}= \bf{1}} \nonumber
\end{align}
where \eqref{probY_S_ineq} follows since $\Prob{\bigcap_{R\in \cR}\rmE(R)} = (1-\delta)^{\abs{[\cR]}}$
}

\qed
}

\subsection{Ungar's\texorpdfstring{~\cite{ungar1960cutoff}}{} Adaptive Algorithm}\label{sec:FKW_algo}

Let $\zeta \defeq 1-\delta$. Below we analyze the expected number of tests required in Algorithm~\ref{alg:adaptive_GT}

\begin{algorithm}
	\KwData{$n$ objects such that each object is defective independently with probability $\delta$}
	\KwResult{the defective set, $D$}
	$D=\emptyset$ \\
	\uIf{$\delta \geq \frac{3-\sqrt5}2$}{
		Test each object individually
	}
	\uElse{
		Partition the $n$ items into $n/2$ disjoint pairs. \\
		\While{there exist untested pairs $\set{i_1, i_2}$}{
			$Test =  X_{i_1} \vee X_{i_1}$\; \label{test1}
			\uIf{$Test=1$}{
				$Test =  X_{i_1}$ \;\label{test2}
			}
			\uElseIf{$Test=1$}{
				$D = D\cup \set{i_1}$\;
				$Test =  X_{i_2}$ \;\label{test3}
			}
			\uElseIf{$Test=1$}{
				$D = D\cup \set{i_2}$;
			}
		}
	}
 \caption{Adaptive Algorithm for Group Testing}\label{alg:adaptive_GT}
\end{algorithm}

Algorithm~\ref{alg:adaptive_GT} conducts tests in lines \ref{test1}, \ref{test2}, \ref{test3} :
\begin{enumerate}
	\item The first test in line \ref{test1} is always performed,
	\item The second test in line \ref{test2} is performed iff $X_{i_1}\vee X_{i_2} = 1$.
	\item The third test in line \ref{test3} is performed iff $X_{i_1} = 1$.
\end{enumerate}
Thus the expected number of tests performed is 
\begin{equation}\label{adaptive_tests_appendix}
	n \min\left \{1,\frac{1}2 \br*{1 + (1-\zeta^2) + (1-\zeta)}\right \} = n \min\left \{1,\frac{1}2 \br*{3 - \zeta - \zeta^2}\right \}
\end{equation}

\bibliographystyle{IEEEtranS}
\bibliography{GTsummer}

% Generated by IEEEtranS.bst, version: 1.14 (2015/08/26)
\begin{thebibliography}{10}
\providecommand{\url}[1]{#1}
\csname url@samestyle\endcsname
\providecommand{\newblock}{\relax}
\providecommand{\bibinfo}[2]{#2}
\providecommand{\BIBentrySTDinterwordspacing}{\spaceskip=0pt\relax}
\providecommand{\BIBentryALTinterwordstretchfactor}{4}
\providecommand{\BIBentryALTinterwordspacing}{\spaceskip=\fontdimen2\font plus
\BIBentryALTinterwordstretchfactor\fontdimen3\font minus
  \fontdimen4\font\relax}
\providecommand{\BIBforeignlanguage}[2]{{%
\expandafter\ifx\csname l@#1\endcsname\relax
\typeout{** WARNING: IEEEtranS.bst: No hyphenation pattern has been}%
\typeout{** loaded for the language `#1'. Using the pattern for}%
\typeout{** the default language instead.}%
\else
\language=\csname l@#1\endcsname
\fi
#2}}
\providecommand{\BIBdecl}{\relax}
\BIBdecl

\bibitem{aldridge2016improved}
M.~Aldridge, O.~Johnson, and J.~Scarlett, ``{Improved Group Testing Rates with
  Constant Column Weight Designs},'' in \emph{IEEE Int Symp Info}, 2016, pp.
  1381--1385.

\bibitem{beimel2011secret}
A.~Beimel, ``{Secret-Sharing Schemes: A Survey},'' \emph{IWCC}, vol. 6639, pp.
  11--46, 2011.

\bibitem{berger1984random}
T.~Berger, N.~Mehravari, D.~Towsley, and J.~Wolf, ``{Random Multiple-access
  Communication and Group Testing},'' \emph{IEEE T Comm}, vol.~32, no.~7, pp.
  769--779, 1984.

\bibitem{chan2014non}
C.~L. Chan, S.~Jaggi, V.~Saligrama, and S.~Agnihotri, ``{Non-adaptive Group
  Testing: Explicit Bounds and Novel Algorithms},'' \emph{IEEE T Inform
  Theory}, vol.~60, no.~5, pp. 3019--3035, 2014.

\bibitem{damaschke2006threshold}
P.~Damaschke, ``{Threshold Group Testing},'' in \emph{General theory of
  information transfer and combinatorics}.\hskip 1em plus 0.5em minus
  0.4em\relax Springer, 2006, pp. 707--718.

\bibitem{damaschke2010bounds}
P.~Damaschke and A.~S. Muhammad, ``{Bounds for Nonadaptive Group Tests to
  Estimate the Amount of Defectives},'' in \emph{Intl Conf Combinatorial
  Optimization and Appl}.\hskip 1em plus 0.5em minus 0.4em\relax Springer,
  2010, pp. 117--130.

\bibitem{delsarte1973algebraic}
P.~Delsarte, ``{An Algebraic Approach to the Association Achemes of Coding
  Theory},'' \emph{Philips Res. Reports Suppls.}, vol.~10, 1973.

\bibitem{dorfman1943detection}
R.~Dorfman, ``{The Detection of Defective Members of Large Populations},''
  \emph{Annals of Mathematical Statistics}, vol.~14, no.~4, pp. 436--440, 1943.

\bibitem{du2000combinatorial}
D.-Z. Du and F.~K. Hwang, \emph{{Combinatorial Group Testing and its
  Applications}}.\hskip 1em plus 0.5em minus 0.4em\relax World Scientific,
  2000.

\bibitem{d1982bounds}
A.~G. D'yachkov and V.~V. Rykov, ``{Bounds on the Length of Disjunctive
  Codes},'' \emph{Problemy Peredachi Informatsii}, vol.~18, no.~3, pp. 7--13,
  1982.

\bibitem{C}
P.~Fischer, N.~Klasner, and I.~Wegener, ``{On the Cut-off Point for
  Combinatorial Group Testing},'' \emph{Discrete Applied Mathematics}, vol.~91,
  no. 1-3, pp. 83--92, 1999.

\bibitem{Han78}
T.~S. Han, ``{Nonnegative Entropy Measures of Multivariate Symmetric
  Correlations},'' \emph{Information and Control}, vol.~36, no.~2, pp.
  133--156, 1978.

\bibitem{E}
M.~Hu, F.~Hwang, and J.~K. Wang, ``{A Boundary Problem for Group Testing},''
  \emph{SIAM J Algebraic Discrete Methods}, vol.~2, no.~2, pp. 81--87, 1981.

\bibitem{A}
F.~Hwang, ``{A Method for Detecting all Defective Members in a Population by
  Group Testing},'' \emph{American Statistical Association}, vol.~67, no. 339,
  pp. 605--608, 1972.

\bibitem{hwang1987non}
F.~Hwang and V.~S{\'o}s, ``{Non-adaptive Hypergeometric Group Testing},''
  \emph{Studia Sci. Math. Hungar}, vol.~22, pp. 257--263, 1987.

\bibitem{kailath2000linear}
T.~Kailath, A.~H. Sayed, and B.~Hassibi, \emph{{Linear estimation}}.\hskip 1em
  plus 0.5em minus 0.4em\relax Prentice Hall Upper Saddle River, NJ, 2000,
  vol.~1.

\bibitem{macula2004group}
A.~J. Macula and L.~J. Popyack, ``{A group testing method for finding patterns
  in data},'' \emph{Discrete Applied Mathematics}, vol. 144, no.~1, pp.
  149--157, 2004.

\bibitem{madimanTetali09}
M.~Madiman and P.~Tetali, ``{Information Inequalities for Joint Distributions,
  With Interpretations and Applications},'' \emph{IEEE T Inform Theory},
  vol.~56, no.~6, pp. 2699--2713, 2010.

\bibitem{mazumdar2016nonadaptive}
A.~Mazumdar, ``{Nonadaptive Group Testing with Random Set of Defectives},''
  \emph{IEEE T Inform Theory}, vol.~62, no.~12, pp. 7522--7531, 2016.

\bibitem{mceliece1977new}
R.~McEliece, E.~Rodemich, H.~Rumsey, and L.~Welch, ``{New upper bounds on the
  rate of a code via the Delsarte-MacWilliams inequalities},'' \emph{IEEE T
  Inform Theory}, vol.~23, no.~2, pp. 157--166, 1977.

\bibitem{ngo2000survey}
H.~Q. Ngo and D.-Z. Du, ``{A Survey on Combinatorial Group Testing Algorithms
  with Applications to DNA Library Screening},'' \emph{Discrete Math Problems
  with Medical Appl}, vol.~55, pp. 171--182, 2000.

\bibitem{reeves2013approximate}
G.~Reeves and M.~C. Gastpar, ``{Approximate Sparsity Pattern Recovery:
  Information-theoretic Lower Bounds},'' \emph{IEEE T Inform Theory}, vol.~59,
  no.~6, pp. 3451--3465, 2013.

\bibitem{riccio2000sharper}
L.~Riccio and C.~J. Colbourn, ``{Sharper Bounds in Adaptive Group Testing},''
  \emph{Taiwanese Journal of Mathematics}, pp. 669--673, 2000.

\bibitem{stanley2012numerative}
R.~P. Stanley, ``{Enumerative Combinatorics. Vol. 1, Cambridge Studies in
  Advanced Mathematics},'' 2012.

\bibitem{ungar1960cutoff}
P.~Ungar, ``{The Cutoff Point for Group Testing},'' \emph{Communications on
  Pure and Applied Math}, vol.~13, no.~1, pp. 49--54, 1960.

\bibitem{wadayama2017nonadaptive}
T.~Wadayama, ``{Nonadaptive Group Testing Based On Sparse Pooling Graphs},''
  \emph{IEEE T Inform Theory}, vol.~63, no.~3, pp. 1525--1534, 2017.

\bibitem{wang2016data}
I.-H. Wang, S.-L. Huang, K.-Y. Lee, and K.-C. Chen, ``{Data Extraction via
  Histogram and Arithmetic Mean Queries: Fundamental Limits and Algorithms},''
  in \emph{IEEE Int Symp Info}, 2016, pp. 1386--1390.

\bibitem{wolf1985born}
J.~Wolf, ``{Born Again Group Testing: Multiaccess Communications},'' \emph{IEEE
  T Inform Theory}, vol.~31, no.~2, pp. 185--191, 1985.

\bibitem{zhang1997non}
Z.~Zhang and R.~W. Yeung, ``{A non-Shannon-type Conditional Inequality of
  Information Quantities},'' \emph{IEEE T Inform Theory}, vol.~43, no.~6, pp.
  1982--1986, 1997.

\bibitem{zhang1998characterization}
------, ``{On Characterization of Entropy Function via Information
  Inequalities},'' \emph{IEEE T Inform Theory}, vol.~44, no.~4, pp. 1440--1452,
  1998.

\bibitem{zhigljavsky2003probabilistic}
A.~Zhigljavsky, ``{Probabilistic existence theorems in group testing},''
  \emph{J Statistical Planning and Inference}, vol. 115, no.~1, pp. 1--43,
  2003.

\end{thebibliography}
\end{document}